\documentclass[journal]{IEEEtran}
\usepackage{setspace}
\usepackage{fancyhdr}
\usepackage{graphicx}
\usepackage{color}
\usepackage{placeins}
\usepackage{float}
\usepackage{tabularx,colortbl}
\usepackage{amsmath}
\usepackage{amsfonts}
\usepackage{amssymb}
\usepackage{amsthm}
\usepackage[numbers,sort&compress]{natbib}
\usepackage[linesnumbered,boxed,ruled,commentsnumbered]{algorithm2e}
\usepackage{bm}
\usepackage{epstopdf}

\begin{document}
%
\title{Utility-Energy Efficiency Oriented User Association with Power Control in Heterogeneous Networks}
\author{
\IEEEauthorblockN{Xietian~Huang, Wei~Xu, Hong~Shen,  Hua~Zhang, and Xiaohu~You}\\
\vspace{-0.8cm}
\thanks{
Manuscript received July 30, 2017; revised November 17, 2017 and December 29, 2017; accepted January 05, 2018. This work of was supported in part by the Hong Kong, Macao and Taiwan Science \& Technology Cooperation Program of China (2016YFE0123100), NSFC (grant nos. 61471114, 61501110, 61571118), the Six Talent Peaks project in Jiangsu Province under GDZB-005, the Open Research Fund of the State Key Lab of ISN under ISN18-03, and the NSF of Jiangsu Province under BK20150635. The editor coordinating the review of this paper and approving it for publication was M. Assaad.
\textit{(Corresponding authors: Wei Xu, Hong Shen.)}

All authors are with the National Mobile Communications Research Laboratory (NCRL), Southeast University, Nanjing 210096, China. X. Huang and W. Xu are also with the State Key Laboratory of Integrated Services Networks, Xidian University, Xi'an 710071, China (e-mail: xthuang@seu.edu.cn; wxu@seu.edu.cn).

}

}



%


\maketitle

\newtheorem{mylemma}{Lemma}
\newtheorem{mytheorem}{Theorem}
\newtheorem{mypro}{Proposition}
\begin{abstract}
This letter investigates optimizing utility-energy efficiency (UEE), defined as the achieved network utility when consuming a unit of power, rather than a typical energy efficiency metric, in a heterogeneous network (HetNet). To tackle the nonconvexity of the problem due to integer constraints and coupled variables, we devise an alternating optimization algorithm. It applies Lagrangian dual analysis with auxiliary variables, which successfully transforms each subproblem to a convex one with efficient solutions. The proposed algorithm is consequently guaranteed to converge to a local optimum with noticeable performance gain via simulation verifications.
\end{abstract}

\begin{IEEEkeywords}
Utility-energy efficiency (UEE), user association, power control, heterogeneous network (HetNet), load balancing.
\vspace{-0.5cm}
\end{IEEEkeywords}

%
\IEEEpeerreviewmaketitle
\section{Introduction}
Heterogeneous network (HetNet) has recently become a research focus as an effective technology to improve spectrum efficiency (SE) \cite{Damnianovic2011A}. Different from traditional network, HetNet equips low-power small-cell base stations (SBSs) besides macro-cell base station (MBS), which makes the deployment more flexible while also comes with several challenges.

User association is a significant issue that needs to be reconsidered. Since the transmit power of MBS is much higher than that of SBS, most users tend to be associated with MBS via classical association schemes, leading to unbalanced load. In this way, MBS users are unlikely to reach high rate because they have to share the resource of MBS \cite{Ye2012User}. While the whole network can benefit by transferring some MBS users to lightly-loaded SBSs. Transmit power control is a further resource allocation problem directly related to user association \cite{Shen2014Distributed}. A proper setting of power can decrease the interference between different tiers of BSs, which strongly influences the achievable rate. Early works like \cite{Hanly1995An} have revealed the benefits of joint optimization of user association and power control.

With the increasing energy costs of wireless networks, energy efficiency (EE) has become an important metric in 5G \cite{Cui2016Optimal}. Generally, EE design maximizes the achieved rate evaluated by per unit power consumption. In practice, however, one ultimate metric that matters is the network utility instead of the raw theoretical rate. Therefore, we propose a utility-energy efficiency (UEE) metric. By exploiting the popular log-utility model, we formulate the joint optimization of association and power control by UEE maximization. Since the original problem is nonconvex, we propose an iterative algorithm by solving the association and power control problem alternately. Numerical results show that the proposed algorithm outperforms existing methods in terms of various metrics including UEE.

\section{System Model}\label{sec:system_model}
Consider a downlink HetNet consisting of $N_m$ MBSs and $N_s$ SBSs. Let $\mathbb{B}$ be the set of all BSs. $\mathbb{U}$ denotes the set of all users with size $N_u$. The received signal of user $i$ is
\vspace{-0.2cm}
\begin{equation}
y_{i}=\sum_{j\in \mathbb{B}}h_{ij}\sqrt{p_{j}}s_{j}+n_{i},\forall i\in \mathbb{U}, j\in \mathbb{B}
\vspace{-0.2cm}
\end{equation}
where $s_{j}$ is the transmit signal, $h_{ij}$ is the flat-fading channel gain, $n_{i}$ is the additive zero-mean Gaussian noise with variance $\sigma^{2}$, and $p_{j}$ denotes the transmit power. Assume that the BS has the global channel state information (CSI) for optimization and the channels vary slowly. Thus we can obtain the received signal-to-interference-plus-noise ratio (SINR) as
\begin{equation}\label{eq:SINR}
 \text{SINR}_{ij}=\frac{h_{ij}p_{j}}{\sum_{q\neq j}h_{iq}p_{q}+\sigma^{2} },\forall i\in \mathbb{U}, j\in \mathbb{B}.
\end{equation}

Denote binary variables $\{x_{ij}\}$ as the indicator of the association. If user $i$ is associated with BS $j$, then $x_{ij}\!=\!1$; otherwise $x_{ij}\!=\!0$. Although letting each user associate with more than one BS can avoid the combinatorial assignment issue \cite{Chien2016Joint}, it becomes difficult to implement multiple-BS association in practice \cite{Ye2012User}, when considering the synchronization and control signaling. Thus, we assume that each user can be associated with only one BS at a time. Denote $k_{j}$ as the number of users associated with BS $j$, i.e., $k_{j}\!=\!\sum_{i\in \mathbb{U}} x_{ij}$. The system bandwidth $W$ is reused by all BSs, and users associated with the same BS share the frequency resource. Our proposed algorithm considers a simple uniform resource allocation, thus the resource allocated to each user is $W/k_{j}$. The joint optimization of association and bandwidth allocation may achieve better performance at the expense of increased complexity, which is of great interest for further considerations.

According to the Shannon's formula, when user $i$ is associated with BS $j$, the achievable rate is characterized as
\begin{equation}\label{eq:rate}
 c_{ij}=(W/k_{j})\log(1+\text{SINR}_{ij}),\forall i\in \mathbb{U}, j\in \mathbb{B}.
\end{equation}

In downlink networks, EE is defined as the ratio of the sum rate to the total power consumption. However, maximizing system EE may result in extremely unfair throughput allocation. In practice, the network utility is more meaningful to subscribers. To preserve some degree of fairness, we consider optimizing the network UEE, i.e., maximizing the ratio of the sum utility rate to the power consumption. The logarithmic function is a typical utility function which is proven to achieve tradeoff between network throughput and fairness. Now we can formulate the UEE optimization problem by jointly finding the optimal association and power control strategies. It follows
\vspace{-0.2cm}
\begin{subequations}\label{eq:pro_X_P}
\begin{align}
\mathop{\max}_{\bm{X},\bm{k},\bm{p}}\quad\!\!
&\frac{\sum\limits_{i\in \mathbb{U}}\sum\limits_{j\in \mathbb{B}} x_{ij}\log c_{ij}}{\sum\limits_{j\in \mathbb{B}}p_{j}+P_c}\\
\textrm{s.t.}\quad\!\!
&0\leq p_{j}\leq P_{j}^{m},\; \forall j\in \mathbb{B}\\
&\small{\sum\limits_{j\in \mathbb{B}}} x_{ij} = 1, \; \forall i\in \mathbb{U}\\
&\small{\sum\limits_{i\in \mathbb{U}}} x_{ij} = k_{j}, \quad \small{\sum\limits_{j\in \mathbb{B}}} k_{j} = N_u\\
& x_{ij}\in \{0,1\}, \; \forall i\in \mathbb{U}, j\in \mathbb{B}
\end{align}
\end{subequations}
where $\bm{X}\!=\![x_{ij}]_{N_{u}\times (N_{s}+N_{m})}$ is the association matrix, $\bm{k}\!=\!\{k_{j}\}_{j\in \mathbb{B}}$ denotes the BS load, $\bm{p}\!=\!\{p_{j}\}_{j\in \mathbb{B}}$ is the transmit power vector, $P_{c}$ is a constant denoting the circuit power consumption, and $P_{j}^{m}$ is the maximum power constraint.

\section{UEE Oriented Resource Allocation}
Problem \eqref{eq:pro_X_P} is a typical nonlinear fractional programming and can be equivalently transformed via parametric programming \cite{Xu2015Robust}. Define the parametric subtractive problem as
\begin{equation}\label{eq:pro_eta}
 F(\eta)= \max\limits_{\bm{X},\bm{k},\bm{p}} \quad \sum\limits_{i\in \mathbb{U}}\sum\limits_{j\in \mathbb{B}} x_{ij}\log c_{ij}-\eta(\sum\limits_{j\in \mathbb{B}}p_{j}+P_c).
\end{equation}
The solution $\{\bm{X}^*,\bm{k}^*,\bm{p}^*\}$ to \eqref{eq:pro_X_P} is also optimal for \eqref{eq:pro_eta} for a certain $\eta^{*}\geq 0$ that satisfies $F(\eta^{*})= 0$. The optimal value of \eqref{eq:pro_X_P} is equal to $\eta^{*}$. For fixed $\eta$, problem \eqref{eq:pro_eta} has the form as
\begin{equation}\label{eq:pro_X_P_eta}
\mathop{\max}_{\bm{X},\bm{k},\bm{p}}\ \sum\limits_{i\in \mathbb{U}}\sum\limits_{j\in \mathbb{B}} x_{ij}\log c_{ij}-\eta \sum\limits_{j\in \mathbb{B}}p_{j},\; \textrm{s.t.}\;(4\text{b})-(4\text{e}).
\end{equation}

Problem \eqref{eq:pro_X_P_eta} is still a nonconvex mixed-integer problem and the optimum is difficult to find. Since the SINR only depends on transmit power, this problem will be relatively tractable if $\bm{p}$ is temporarily fixed. Thus, we first consider the association problem under fixed $\bm{p}$, and then deal with the power control with fixed $\bm{X}$. Joint optimization is conducted alternately.

\subsection{Optimal User Association with Fixed Power}
Given $\bm{p}$, and introducing $m_{ij}\triangleq \log\left(W\log\left(1+\text{SINR}_{ij}\right)\right)$, the user association problem can be equivalently rewritten as
\begin{equation}\label{eq:problem_X}
\mathop{\max}_{\bm{X},\bm{k}}\ \sum\limits_{i\in \mathbb{U}}\sum\limits_{j\in \mathbb{B}} x_{ij}m_{ij}-\sum\limits_{j\in \mathbb{B}}k_{j}\log k_{j},\; \textrm{s.t.}\;(4\text{c})-(4\text{e}).
\end{equation}

Problem \eqref{eq:problem_X} is generally complicated due to the binary-valued constraint (4e). We in the following first adopt the fractional user association relaxation, where $x_{ij}$ can take any real value in [0,1]. Later we will show that fortunately the relaxed problem generates optimal $x_{ij}^{*}$ as integers which is thus guaranteed as the optimum of original problem \eqref{eq:problem_X}.

The relaxed user association problem equals:
\begin{subequations}\label{eq:relaxation}
\begin{align}
\mathop{\max}_{\bm{X},\bm{k}}\quad\!\!
&\sum\limits_{i\in \mathbb{U}}\sum\limits_{j\in \mathbb{B}} x_{ij}m_{ij}-\sum\limits_{j\in \mathbb{B}}k_{j}\log k_{j}\\
\textrm{s.t.}\quad\!\!
&(4\text{c}),(4\text{d}),0\leq x_{ij}\leq 1.
\end{align}
\end{subequations}
\begin{mylemma}
Problem \eqref{eq:relaxation} is a convex problem.
\end{mylemma}
\begin{proof}
It is easy to prove the concavity of the objective function by checking its Hessian matrix which is diagonal with nonpositive elements, and all constraints are linear.
\end{proof}
To deal with \eqref{eq:problem_X}, a typical way is to solve convex problem \eqref{eq:relaxation} by interior-point method \cite{Boyd2004Convex} and then conduct rounding on $\bm{X}$. The optimality however may not be preserved in theory. Considering that the dimension of the variables in \eqref{eq:relaxation} is $(N_u\!+\!1)(N_m\!+\!N_s)$, the complexity of solving \eqref{eq:relaxation} by the standard interior-point method is $\mathcal{O}(N_u^3(N_m\!+\!N_s)^3)$ [8, p. 487, 569], which is generally high. Here, we adopt the Lagrangian dual decomposition analysis \cite{Palomar2006A} for achieving low-complexity and optimum guaranteed solutions. Introducing dual variables $\bm{\mu}\!=\!\{\mu_{j}\}_{j\in \mathbb{B}}$ and $\nu$, the Lagrangian function of \eqref{eq:relaxation} is
\vspace{-0.1cm}
\begin{align}\label{eq:lagfun}
  L(\bm{X},\bm{k},\bm{\mu},\nu)\!=\!&\sum\limits_{i\in \mathbb{U}}\sum\limits_{j\in \mathbb{B}} x_{ij}m_{ij}\!-\!\sum\limits_{j\in \mathbb{B}}k_{j}\log k_{j}\\
                 &\!-\!\sum\limits_{j\in \mathbb{B}}\mu_{j}(\sum\limits_{i\in \mathbb{U}}x_{ij}\!-\!k_{j})\!-\!\nu(\sum\limits_{j\in \mathbb{B}}k_{j}\!-\!N_u).\nonumber
\end{align}

It is readily to find that the convex problem in \eqref{eq:relaxation} satisfies the Slater's condition, which means that strong duality holds \cite{Boyd2004Convex}. Therefore, primal problem \eqref{eq:relaxation} can be equivalently solved by solving the dual problem.

The optimal $k_{j}^{*}$ can be obtained by letting $\frac{\partial L(\cdot)}{\partial k_{j}}\!=\!0$. Then by rewriting the Lagrangian function and removing irrelevant items, the optimization problem on $\bm{X}$ is simplified as
\vspace{-0.1cm}
\begin{equation}
\max\limits_{\bm{X}}\ \sum\limits_{i\in \mathbb{U}}\sum\limits_{j\in \mathbb{B}}x_{ij}(m_{ij}-\mu_{j}),\;\textrm{s.t.}\;(4\text{c}),0\leq x_{ij}\leq 1.
\end{equation}
The objective function is upper bounded by
\vspace{-0.1cm}
\begin{equation}\label{eq:analysis}
\sum\limits_{i\in \mathbb{U}}\sum\limits_{j\in \mathbb{B}}x_{ij}(m_{ij}-\mu_{j})\leq \sum\limits_{i\in \mathbb{U}}\max\limits_{j\in \mathbb{B}}(m_{ij}-\mu_{j}).
\end{equation}
If there exists one feasible $x_{ij}$ making (11) achieve with equality, it is exactly the optimal solution to (10). Intuitively, we find a feasible solution as given by
\vspace{-0.1cm}
\begin{equation}\label{eq:opt_X}
x_{ij}^{*}\!=\!
\begin{cases}
1, &\text{if}\; j \!=\! j^{(i)}\\
0, &\text{if}\; j \!\neq\! j^{(i)},
\end{cases}
\text{where}\;j^{(i)}\!=\!\arg\;\max\limits_{q\in \mathbb{B}}(m_{iq}\!-\!\mu_{q}).
\end{equation}
We find that although the binary constraint is relaxed in \eqref{eq:relaxation}, the optimal $x_{ij}^{*}$ is fortunately either 0 or 1 which exactly satisfies constraints in \eqref{eq:problem_X}. Thus, the optimum obtained by solving dual problem is in fact optimal to original problem \eqref{eq:problem_X}.

The dual variables are iteratively updated as
\vspace{-0.1cm}
\begin{equation}\label{eq:opt_u}
\mu_{j}^{(t+1)}=\mu_{j}^{(t)}-\delta^{(t)}({\rm e}^{\mu_{j}^{(t)}-\nu^{(t)}-1}-\sum\limits_{i\in \mathbb{U}}x_{ij}^{(t)})
\end{equation}
\vspace{-0.1cm}
\begin{equation}\label{eq:opt_v}
\nu^{(t)}=(\log \sum\limits_{j\in \mathbb{B}}{\rm e}^{\mu_{j}^{(t)}-1})/N_u
\vspace{-0.2cm}
\end{equation}
where $\delta^{(t)}$ is the step size.

The Lagrangian dual method follows our engineering intuitions. Regard $\mu_j$ as the price of BS $j$ and $m_{ij}$ as the utility rate if user $i$ is associated with BS $j$. When choosing the BS, each user considers maximizing the utility rate minus the price according to \eqref{eq:opt_X}, while each BS updates its price to balance load via \eqref{eq:opt_u}, which indicates the law of supply and demand. If the service demand of users $\sum_{i\in \mathbb{U}}x_{ij}^{(t)}$ is larger than the total supply amount ${\rm e}^{\mu_{j}^{(t)}-\nu^{(t)}-1}$ of BS $j$, it increases the price.

\subsection{Power Control Method}\label{sec:power control}
Given the optimized association, we then focus on power optimization. With fixed $\bm{X}$ and introducing auxiliary variables $\bm{\lambda}=\{\lambda_{ij}\}_{i\in \mathbb{U},j\in \mathbb{B}}$, the power control problem equals to
\begin{subequations}\label{eq:pro_P_lam1}
\begin{align}
\mathop{\max}_{\bm{p},\bm{\lambda}}\quad\!\!
&\sum\limits_{i\in \mathbb{U}}\sum\limits_{j\in \mathbb{B}} x_{ij}\log\left(\log\left(1+\lambda_{ij}\right)\right)-\eta\sum\limits_{j\in \mathbb{B}}p_{j}\\
\textrm{s.t.}\quad\!\!
&(4b),\ \frac{h_{ij}p_{j}}{\sum_{q\neq j}h_{iq}p_{q}+\sigma^{2}}\geq \lambda_{ij},\; \forall i\in \mathbb{U}, j\in \mathbb{B}.
\end{align}
\end{subequations}
Note that introducing $\bm{\lambda}$ helps to transform the objective function into concave. Since each user is associated with one BS, there exists only one $j$ such that $x_{ij}\!=\!1$ for user $i$. Let $\mathbb{U}_j\!=\!\{i\in \mathbb{U}|x_{ij}\!=\!1\}$ denote the set of users associated with BS $j$. For convenience, we use $\bm{\lambda^{'}}\!=\!\{\lambda_{i}\}_{i\in \mathbb{U}_j}$ to replace $\bm{\lambda}$. Then problem \eqref{eq:pro_P_lam1} can be further reformulated as
\begin{subequations}\label{eq:pro_P_lam2}
\begin{align}
\mathop{\max}_{\bm{p},\bm{\lambda^{'}}}\quad\!\!
&\sum\limits_{j\in \mathbb{B}}\sum\limits_{i\in \mathbb{U}_j} \log\left(\log\left(1+\lambda_{i}\right)\right)-\eta\sum\limits_{j\in \mathbb{B}}p_{j}\\
\textrm{s.t.}\quad\!\!
&(4b),\ \frac{h_{ij}p_{j}}{\sum_{q\neq j}h_{iq}p_{q}+\sigma^{2}}\geq \lambda_{i},\; \forall j\in \mathbb{B}, i\in \mathbb{U}_j.
\end{align}
\end{subequations}

Introducing auxiliary variables ${\rm e}^{\rho_{j}}\triangleq p_j$ and ${\rm e}^{\theta_{i}}\triangleq \lambda_{i}$, the nonconvex constraints in (16) thus become convex as
\begin{equation}
 {\rm e}^{\theta_{i}-\rho_{j}+\beta_{i}}+\sum_{q\neq j}{\rm e}^{\theta_{i}-\rho_{j}+\rho_{q}+\gamma_{iq}} \leq 1
 \vspace{-0.2cm}
\end{equation}
where $\beta_{i} \triangleq \log\left( \frac{\sigma^{2}}{h_{ij}}\right)$ and $\gamma_{iq} \triangleq \log\left(\frac{h_{iq}}{h_{ij}}\right)$ are constants. Further, by defining $\omega_i \triangleq \theta_{i}-\rho_{j}+\beta_{i}$ and $s_{ijq} \triangleq \theta_{i}-\rho_{j}+\rho_{q}+\gamma_{iq}$, problem \eqref{eq:pro_P_lam2} is reformulated as
\begin{subequations}\label{eq:pro_rho_theta}
\begin{align}
\mathop{\max}_{\bm{\rho},\bm{\theta},\bm{\omega},\bm{s}}\quad\!\!
&\sum\limits_{j\in \mathbb{B}} \sum\limits_{i\in \mathbb{U}_j} \log\left(\log\left(1+{\rm e}^{\theta_{i}}\right)\right)-\eta\sum\limits_{j\in \mathbb{B}}{\rm e}^{\rho_{j}}\\
\textrm{s.t.}\quad\!\!
&\rho_{j}\leq \log \left(P_{j}^{m}\right), \; \forall j\in \mathbb{B} \\
&{\rm e}^{\omega_i}\!+\!\sum_{q\neq j}{\rm e}^{s_{iq}} \leq 1, \; \forall j\in \mathbb{B}, i\in \mathbb{U}_j \\
&\omega_i\!=\!\theta_{i}\!-\!\rho_{j}\!+\!\beta_{i}, \; \forall j\in \mathbb{B}, i\in \mathbb{U}_j \\
&s_{ijq}\!=\!\theta_{i}\!-\!\rho_{j}\!+\!\rho_{q}\!+\!\gamma_{iq}, \, \forall j\in \mathbb{B}, i\in \mathbb{U}_j, q \neq j
\end{align}
\end{subequations}
where $\bm{\rho}=\{\rho_{j}\}_{j\in \mathbb{B}}$, $\bm{\theta}=\{\theta_{i}\}_{i\in \mathbb{U}_j}$, $\bm{\omega}=\{\omega_{i}\}_{i\in \mathbb{U}_j}$, and $\bm{s}=\{s_{ijq}\}_{j\in \mathbb{B}, i\in \mathbb{U}_j, q\neq j\in \mathbb{B}}$.

\IncMargin{1em} 
\begin{algorithm}
    Initialize the parameter $\eta=0$, and a small $\varsigma>0$\;
    \Repeat{$\varsigma^{*}\leq \varsigma$}
            {Initialize any feasible $\bm{p}$\;
                \Repeat{\text{convergence}}
                {Solve association problem (7) with fixed $\bm{p}$\;
                 Solve power control problem (15) with fixed $\bm{X}$\;
                }
                Calculate $\varsigma^{*}\!=\!\sum\limits_{i\in \mathbb{U}}\sum\limits_{j\in \mathbb{B}} x_{ij}\log c_{ij}\!-\!\eta(\sum\limits_{j\in \mathbb{B}}p_{j}\!+\!P_c)$\;
                Update $\eta\!=\!(\sum\limits_{i\in \mathbb{U}}\sum\limits_{j\in \mathbb{B}} x_{ij}\log c_{ij})/(\sum\limits_{j\in \mathbb{B}}p_{j}\!+\!P_c)$\;
            }
   Output optimal $\bm{X}^*$ and $\bm{p}^*$\;
    \caption{IUAPC Algorithm\label{al1}}
\end{algorithm}
\DecMargin{1em}

Problem \eqref{eq:pro_rho_theta} is convex because the objective function is concave and the constraints are either linear or convex. Thus, the globally optimal solution can be obtained by, e.g., the interior-point method. We stress that efficient solutions to this specific problem can be acquired by exploiting the dual method \cite{Yang2017Fair}. Details are given in the Appendix.

Thus far, we are ready to present the iterative user association and power control (IUAPC) algorithm in \textbf{Algorithm 1}. Since the objective value increases with global optimum found in each iteration and it has a finite upper bound, the overall iterative algorithm is guaranteed to converge \cite{Xu2015Robust}\cite{Wang2012Joint}. Besides, the solution can be a local optimum following the proof in \cite{Xu2015Robust}. For our proposed algorithm, the total complexity amounts to $\mathcal{O}(N_u(N_m\!+\!N_s)(\log_2(1/\epsilon_0)\!+\!N_m\!+\!N_s))$, where $\epsilon_0$ is the accuracy of bisection search. Our proposed algorithm is centralized, which can be more suitable for applications, e.g., in the cloud radio access network (CRAN). The main signaling overhead for implementation is the exchange of channel information which is proportional to $N_u(N_m\!+\!N_s)$.

\begin{table}[t!]
\scriptsize
\setlength{\belowcaptionskip}{-10pt}
  \centering
  \caption{Simulation Parameters}\label{TABLE:parameter}
    \begin{tabular}{|c|c|c|c|}
  \hline
  Channel bandwidth & 10 MHz & Cell radius & 500 m \\
  \hline
  Max power of MBS & -27 dBm/Hz & Number of MBS & 1 \\
  \hline
  Max power of SBS & -47 dBm/Hz & Number of SBS & 3 \\
  \hline
  Circuit constant power & 1 W  & Number of users & 30 \\
  \hline
\end{tabular}
\end{table}

\begin{figure}[t!]
\setlength{\abovecaptionskip}{-5pt}
\centering
\includegraphics[width = 3.0in]{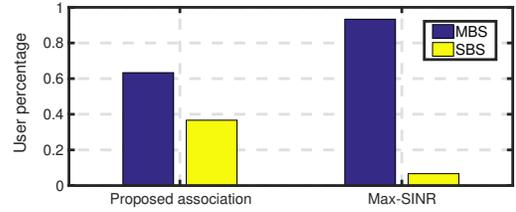}
\caption{The percentage of MBS/SBS users for different association methods.}
\label{Fig:user percentage}
\vspace{-0.3cm}
\end{figure}
\begin{table}[t!]
\scriptsize
\newcommand{\tabincell}[2]{\begin{tabular}{@{}#1@{}}#2\end{tabular}}
\setlength{\belowcaptionskip}{-10pt}
  \centering
  \caption{Average UEE Comparison of Different Algorithms}\label{TABLE:UEE}
  \begin{tabular}{|c|c|c|c|c|}\hline
    & \tabincell{c}{Proposed \\ algorithm} & \tabincell{c}{Max-SINR \\ with power control} & \tabincell{c}{Max-SINR \\ with max power} \\ \hline
UEE & 35.392 & 29.982 & 1.490 \\ \hline
\end{tabular}
\end{table}

\section{Numerical Results}\label{sec:simulation}
In this section, we evaluate the performance of the proposed algorithm via simulation. Consider a downlink 2-tier HetNet with simulation parameters listed in Table \ref{TABLE:parameter}. The pathloss is modelled as $128.1+37.6\log_{10}d$(km), and the shadow fading is log-normally distributed as $\mathcal{N}(0,\sigma^2)$ where $\sigma=8$dB.

Fig. 1 compares the percentage of MBS/SBS users for different association methods. Note that the “SINR” in the Max-SINR association is evaluated with only large-scale channel fading, which is the same as the Max-SNR association in \cite{Chien2016Joint}. For the Max-SINR association, the MBS is overloaded with more than $90\%$ users, while the proposed association achieves balanced load. We also evaluate the performance of our iterative algorithm by comparing with existing methods. In our test, the number of iterations required for the convergence of outer loop is 5, and the optimal $\eta$ is obtained as 35.392. For a fixed $\eta$, the inner loop converges with 2 iterations. Table II compares the average UEE, and Fig. \ref{Fig:CDF2} plots the cumulative distribution function (CDF) of data rates. Our proposed algorithm is shown to perform well in terms of user fairness and throughput. Note that the Max-SINR with power control, which implements our proposed power control with traditional Max-SINR association, outperforms the Max-SINR with max power, proving the significance of power control to improve EE. Moreover, it has been observed that the proposed algorithm achieves near-global optimum via numerical exhaustive search under some small-scale test cases.

\section{Conclusion}\label{sec:conclusion}
In this letter, we jointly considered the user association and power control for HetNets. We formulated a UEE maximization problem and proposed an efficient algorithm. Results demonstrated the validity of our proposed algorithm. Note that the proposed algorithm can be readily extended to a massive MIMO scenario with subtle changes by scaling the SINR. Furthermore, to achieve better performance, the joint optimization with bandwidth allocation is of our future research interest.

\appendix
The dual method is implemented by iteratively solving the primal variables with fixed dual variables, and updating the dual variables with given optimization variables. First, we obtain the Lagrangian function of convex problem \eqref{eq:pro_rho_theta} as
\begin{align*}
\tag{19}
  &L(\bm{\rho},\bm{\theta},\bm{\omega},\bm{s},\bm{a},\bm{b},\bm{\zeta},\bm{\chi})\nonumber\\
                =&\sum\limits_{j\in \mathbb{B}}\sum\limits_{i\in \mathbb{U}_j} \log\left(\log\left(1+{\rm e}^{\theta_{i}}\right)\right)-\eta\sum\limits_{j\in \mathbb{B}}{\rm e}^{\rho_{j}}\nonumber\\
                 &-\sum\limits_{j\in \mathbb{B}}\sum\limits_{i\in \mathbb{U}_j}a_{i}({\rm e}^{\omega_i}+\sum_{q\neq j}{\rm e}^{s_{ijq}}-1)\nonumber\\
                 &-\sum\limits_{j\in \mathbb{B}}b_{j}\left(\rho_{j}-\log P_{j}^{m}\right)-\sum\limits_{j\in \mathbb{B}}\sum\limits_{i\in \mathbb{U}_j}\zeta_{i}(\omega_i-\theta_{i}+\rho_{j}-\beta_{i})\nonumber\\
                 &-\sum\limits_{j\in \mathbb{B}}\sum\limits_{i\in \mathbb{U}_j}\sum\limits_{q\in \mathbb{B}, q\neq j}\chi_{ijq}(s_{ijq}-\theta_{i}+\rho_{j}-\rho_{q}-\gamma_{iq})\nonumber
\end{align*}
where $\bm{a}=\{a_{i}\}_{i\in \mathbb{U}_j}$, $\bm{b}=\{b_{j}\}_{j\in \mathbb{B}}$, $\bm{\zeta}=\{\zeta_{i}\}_{i\in \mathbb{U}_j}$, and $\bm{\chi}=\{\chi_{ijq}\}_{j\in \mathbb{B}, i\in \mathbb{U}_j, q\neq j\in \mathbb{B}}$ are dual variables associated with the corresponding constraints of \eqref{eq:pro_rho_theta}. Then the optimal solution to convex problem \eqref{eq:pro_rho_theta} should satisfy
\setcounter{equation}{19}
\begin{align}
\label{eq:partial_rho}  &\frac{\partial L}{\partial \rho_{j}}\!=\!-\eta {\rm e}^{\rho_{j}}-b_{j}-\sum\limits_{i\in \mathbb{U}_j}\zeta_{i}-\sum\limits_{i\in \mathbb{U}_j}\sum\limits_{q\in \mathbb{B}, q\neq j}\chi_{ijq}\\ \nonumber
&\quad\quad\quad+\sum\limits_{q\in \mathbb{B}, q\neq j}\sum\limits_{i\in \mathbb{U}_q}\chi_{iqj}=0, \\
\label{eq:partial_theta} &\frac{\partial L}{\partial \theta_{i}}\!=\!\frac{{\rm e}^{\theta_{i}}}{\left(1+{\rm e}^{\theta_{i}}\right)\log\left(1+{\rm e}^{\theta_{i}}\right)}+\zeta_{i}+\sum\limits_{q\in \mathbb{B}, q\neq j}\chi_{ijq} =0, \\
\label{eq:partial_omega} &\frac{\partial L}{\partial \omega_{i}}\!=\!-a_{i}{\rm e}^{\omega_i}\!-\!\zeta_{i}\!=\!0,\ \frac{\partial L}{\partial s_{ijq}}\!=\!-a_{i}{\rm e}^{s_{ijq}}\!-\!\chi_{ijq}\!=\!0.
\end{align}

In each iteration, we can calculate the primary variables with given Lagrangian multipliers. Define function $f(x)=\frac{{\rm e}^x}{(1+{\rm e}^x)\log(1+{\rm e}^x)}$, which is strictly decreasing by checking its first-order derivative. Denote $f^{-1}(x)$ as the inverse function of $f(x)$. From \eqref{eq:partial_rho}--\eqref{eq:partial_omega}, we obtain
\begin{align}
& \rho_{j}(t\!+\!1)\!=\!-\log \eta + \nonumber \\
&\quad \quad \quad \; \log(\!-\!b_{j}(t)\!\!-\!\!\sum\limits_{i\in \mathbb{U}_j}\zeta_{i}(t)\!\!-\!\!\sum\limits_{q\neq j}(\sum\limits_{i\in \mathbb{U}_j}\chi_{ijq}(t)\!\!-\!\!\sum\limits_{i\in \mathbb{U}_q}\chi_{iqj}(t))), \nonumber \\
& \theta_{i}(t\!+\!1)\!=\!f^{-1}(-\zeta_{i}(t)\!-\!\sum\limits_{q\in \mathbb{B}, q\neq j}\chi_{ijq}(t)),  \nonumber \\
\label{eq:omega} & \omega_{i}(t\!+\!1)\!=\!\log(-\zeta_{i}(t)/a_{i}(t)),\, s_{ijq}(t\!+\!1)\!=\!\log(-\chi_{ijq}(t)/a_{i}(t)). \nonumber
\end{align}
Alternately, the dual variables are updated as follows:
\begin{equation}
\left\{
         \begin{array}{lr}
         a_{i}(t\!+\!1)\!=\!\left[a_{i}(t)\!+\!\delta^{(t)}\left({\rm e}^{\omega_i(t)}\!+\!\sum_{q\neq j}{\rm e}^{s_{ijq}(t)}\!-\!1\right)\right]^+\nonumber \\
         b_{j}(t\!+\!1)\!=\!\left[b_{j}(t)\!+\!\delta^{(t)}\left(\rho_{j}(t)\!-\!\log P_{j}^{m}\right)\right]^+ \nonumber\\
         \zeta_{i}(t\!+\!1)\!=\!\zeta_{i}(t)\!+\!\delta^{(t)}\left(\omega_i(t)\!-\!\theta_{i}(t)\!+\!\rho_{j}(t)\!-\!\beta_{i}\right)\\
         \chi_{ijq}(t\!\!+\!\!1)\!\!=\!\! \chi_{ijq}(t)\!+\!\delta^{(t)}\left(s_{ijq}(t)\!\!-\!\!\theta_{i}(t)\!\!+\!\!\rho_{j}(t)\!\!-\!\!\rho_{q}(t)\!\!-\!\!\gamma_{iq}\right)\nonumber
         \end{array}
    \right.
\end{equation}
where $\delta^{(t)}$ is the step size and $[x]^+$ returns $\max\{x$, $0\}$.
\begin{figure}[t!]
\setlength{\abovecaptionskip}{-5pt}
\centering
\includegraphics[width = 3.2in]{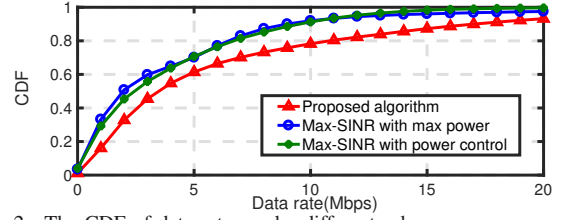}
\caption{The CDF of data rates under different schemes.}
\label{Fig:CDF2}
\vspace{-0.5cm}
\end{figure}

\end{document}